%% file: arxiv-version.tex
\documentclass[a4paper]{article}

\usepackage{microtype}

\usepackage{geometry}
\usepackage{ifthen}
\newif\iflong
\longtrue

\usepackage{xcolor}
\usepackage{amsmath,amssymb,multirow}
\usepackage{mathtools,multicol,tabularx}
\usepackage{mincmc}

\usepackage[standard,amsmath,thmmarks]{ntheorem}

\makeatletter
\renewtheoremstyle{plain}%
{\item[\hskip\labelsep \theorem@headerfont ##1\ ##2 \theorem@separator]}%
{\item[\hskip\labelsep \theorem@headerfont ##1\ ##2\ \normalfont({\sffamily##3})  \theorem@separator]}
\makeatother

\theoremstyle{plain}
\renewtheorem{theorem}{Theorem}
\renewtheorem{remark}[theorem]{Remark}
\renewtheorem{definition}[theorem]{Definition}
\renewtheorem{lemma}[theorem]{Lemma}
\renewtheorem{proposition}[theorem]{Proposition}
\renewtheorem{corollary}[theorem]{Corollary}
\renewtheorem{example}[theorem]{Example}

\newtheorem*{claim}{Claim}

\usepackage{pgf,tikz}
\usetikzlibrary{automata,arrows}

\title{Model checking and validity in propositional and modal inclusion logics\footnote{The second and the last author acknowledges support from Jenny and Antti Wihuri Foundation. The last author is also supported by the grant 292767 of the Academy of Finland. The third author is supported by the DFG grant ME 4279/1-1. }}

\author{Lauri Hella$^1$ \and Antti Kuusisto$^2$ \and Arne Meier$^3$ \and Jonni Virtema$^4$}

\date{\small $^1$ University of Tampere, Finland, \texttt{lauri.hella@uta.fi}\linebreak
$^2$ University of Bremen, Germany, \texttt{antti.j.kuusisto@gmail.com}\linebreak
$^3$ Leibniz Universität Hannover, Germany, \texttt{meier@thi.uni-hannover.de}\linebreak
$^4$ University of Helsinki, Finland, \texttt{jonni.virtema@helsinki.fi}}

\begin{document}

\input{mincmc}
\end{document}

%% file: mincmc.tex
\maketitle

\begin{abstract}
\noindent Propositional and modal inclusion logic are formalisms that belong to the family of logics based on team semantics.
This article investigates the model checking and validity problems of these logics. 
We identify complexity bounds for both problems, covering both lax and strict team semantics. 
By doing so we come close to finalising the programme that ultimately aims to classify the complexities of the basic reasoning problems for modal and propositional dependence, independence, and inclusion logics.
\iflong
\bigskip

\fi
\end{abstract}


%

\section{Introduction}
Team semantics is the mathematical framework of modern logics of dependence and independence, which, unlike Tarski semantics, is not based on singletons as satisfying elements (e.g., first-order assignments or points of a Kripke structure) but on sets of such elements. More precisely, a first-order team is a set of first-order assignments that have the same domain of variables. As a result, a team can be interpreted as a database table, where variables correspond to attributes and assignments to records. Team semantics originates from the work of Hodges \cite{Hodges97c}, where it was shown that Hintikka's IF-logic can be based on a compositional (as opposed to game-theoretic) semantics. In 2007, V\"{a}\"{a}n\"{a}nen \cite{vaananen07} proposed a fresh approach to logics of dependence and independence. V\"{a}\"{a}n\"{a}nen adopted team semantics as a core notion for his \emph{dependence logic}. Dependence logic extends first-order logic by atomic statements such as \emph{the value of variable $x$ is determined by the value of $y$}. Clearly such a statement is not meaningful under a single assignment, however, when evaluated over a team such a statement corresponds precisely to functional dependence of database theory when the team is interpreted as a database table.

Besides functional dependence, there are many other important dependency notions used in fields like statistics and database theory, which give rise to interesting logics based on team semantics.  
The two most widely studied of these new logics are 
\emph{independence logic} of Gr\"{a}del and V\"{a}\"{a}n\"{a}nen \cite{gv13},  and \emph{inclusion logic} of Galliani \cite{Galliani12}. Inclusion logic extends first-order logic by atomic statements of the form $x\subseteq y$, which is satisfied in a team $X$ if any value that appears as a value for $x$ in $X$ also appears as a value of $y$ in $X$. Dependence and independence logics are equal-expressive with existential second-order logic and accordingly capture the complexity class $\NP$ \cite{vaananen07,gv13}. Surprisingly, inclusion logic has the same expressive power as \emph{positive greatest fixed point logic}  $\GFPp$ \cite{GH13}. Since on finite structures, $\GFPp$ coincides with \emph{least fixed point logic} $\LFP$, it follows from the Immermann-Vardi-Theorem that inclusion logic captures the complexity class $\P$ on finite ordered structures. Interestingly under a semantical variant of inclusion logic called \emph{strict semantics} the expressive power of inclusion logic rises to existential second-order logic \cite{ghk13}. Moreover, the fragment of inclusion logic (under strict semantics) in which only $k$ universally quantified variables may occur captures the complexity class $\ntRAM(n^k)$ (i.e., structures that can be recognised by a nondeterministic random access machine in time $\mathcal O(n^k)$)   \cite{hk15}. That being so, indeed, inclusion logic and its fragments have very interesting descriptive complexity theoretic properties. 

In this paper, we study propositional and modal inclusion logic under both the standard semantics (i.e., \emph{lax semantics}) and strict semantics. The research around propositional and modal logics with team semantics has concentrated on classifying the complexity and definability of the related logics. Due to very active research efforts, the complexity and definability landscape of these logics is understood rather well; see the survey of Durand et~al.\ \cite{dukovo16} and the references therein for an overview of the current state of the research. In the context of propositional logic (modal logic, resp.) a team is a set of propositional assignments with a common domain of variables (a subset of the domain a Kripke structure, resp.).  \emph{Extended propositional inclusion logic} (\emph{extended modal inclusion logic}, resp.) extends propositional logic (modal logic, resp.) with \emph{propositional inclusion atoms}  $\varphi\subseteq \psi$, where $\varphi$ and $\psi$ are formulae of propositional logic (modal logic, resp.). The following definability results hold for the standard lax semantics. A class of team pointed Kripke models is definable in extended modal inclusion logic iff $\mK,\emptyset$ is in the class for every model $\mK$, the class is closed under taking unions, and the class is closed under the so-called \emph{team k-bisimulation}, for some finite $k$ \cite{hs15}. From this, a corresponding characterization for extended propositional inclusion logic follows directly. 
In \cite{sv15b,sv16} (global) model definability and frame definability of team based modal logics are studied. It is shown that surprisingly, in both cases, (extended) modal inclusion logic collapses to modal logic.

This paper investigates the complexity of the model checking and the validity problem for propositional and modal inclusion logic. The complexity of the satisfiability problem of modal inclusion logic was studied by Hella et~al.\ \cite{hkmv15}. The study on the validity problem of propositional inclusion logic was initiated by Hannula et~al.\ \cite{hkvv15}, where the focus was on more expressive logics in the propositional setting. Consequently, the current paper directly extends the research effort initiated in these papers. It is important to note that since the logics studied in this paper, are closed under negation, the connection between the satisfiability problem and the validity problem fails.
In \cite{hkvv15} it was shown that, under lax semantics, the validity problem for propositional inclusion logic is $\co\NP$-complete. Here we obtain an identical result for the strict semantics.  However, surprisingly, for model checking the picture looks quite different. We establish that whereas the model checking problem for propositional inclusion logic is $\P$-complete under lax semantics, the problem becomes $\NP$-complete for the strict variant. Also surprisingly, for model checking, we obtain remarkable in the modal setting; modal inclusion logic is $\P$-complete under lax semantics and $\NP$-complete under strict semantics. Nevertheless, for the validity problem, the modal variants are much more complex; we establish $\co\NEXPTIME$-hardness for both strict and lax semantics.
\iflong
For an overview of the results of this paper together with the known complexity results from the literature, see tables \ref{tbl:sat}--\ref{tbl:val} on page \pageref{tbl:sat}. 
\fi

\section{Propositional logics with team semantics}\label{preliminaries}

Let $D$ be a finite, possibly empty set of proposition symbols. A function $s\colon D\to \{0,1\}$ is called an \emph{assignment}. A set $X$ of assignments $s\colon D\to \{0,1\}$ is called a \emph{team}. The set $D$ is the \emph{domain} of $X$. 
We denote by $2^D$ the set of \emph{all assignments} $s\colon D\to \{0,1\}$. If $\vec p =(p_1, \ldots ,p_n)$ is a tuple of propositions  and $s$ is an assignment, we write $s(\vec p)$ for $\left(s(p_1),\dots,s(p_n)\right)$.

Let $\Phi$ be a set of proposition symbols. The syntax of propositional logic $\PL(\Phi)$ is given by the following grammar:
$$
\varphi \ddfn p\mid \neg p \mid (\varphi \wedge \varphi) \mid (\varphi \vee \varphi), \text{ where $p\in\Phi$}.
$$

We denote by $\modelsPL$ the ordinary satisfaction relation of propositional logic defined via assignments in the  standard way. Next we give team semantics for propositional logic. 
\begin{definition}[Lax team semantics]
Let $\Phi$ be a set of atomic propositions and let $X$ be a team. The satisfaction relation $X\models \varphi$ is defined as follows.
\begin{align*}
X\models p  \quad\Leftrightarrow\quad& \forall s\in X: s(p)=1,\\
X\models \neg p \quad\Leftrightarrow\quad& \forall s\in X: s(p)=0.\\
X\models (\varphi\land\psi) \quad\Leftrightarrow\quad& X\models\varphi \text{ and } X\models\psi.\\
X\models (\varphi\lor\psi) \quad\Leftrightarrow\quad& Y\models\varphi \text{ and } 
Z\models\psi,\text{ for some $Y,Z$ such that $Y\cup Z= X$}.
\end{align*}
\end{definition}

The lax team semantics is considered the standard semantics for team-based logics.
In this paper, we also consider a variant of team semantics called the \emph{strict team semantics}.
In strict team semantics, the above clause for disjunction is redefined as follows:
\[
X\models_s (\varphi\lor\psi) \,\Leftrightarrow\, Y\models\varphi \text{ and } 
Z\models\psi,\text{ for some $Y,Z$ such that  $Y\cap Z= \emptyset$ and $Y\cup Z= X$}.
\]
When $\LL$ denotes a team-based propositional logic, we let $\LL_s$ denote the variant of the logic with strict semantics. Moreover, in order to improve readability, for strict semantics we use $\models_s$ instead of $\models$. As a result lax semantics is used unless otherwise specified.
%
%
%
The next proposition shows that the team semantics and the ordinary semantics for propositional logic defined via assignments (denoted by $\modelsPL$) coincide.
\begin{proposition}[\cite{vaananen07}]\label{whatever}
\label{PLflat}
Let $\varphi$ be a formula of propositional logic and let $X$ be a propositional team. Then 
\(
X\models \varphi \;\text{ iff }\; \forall s\in X: s\modelsPL \varphi.
\)
\end{proposition}

The  syntax of \emph{propositional inclusion logic} $\PLinc(\Phi)$ is obtained by extending the syntax of $\PL(\Phi)$ by the grammar rule
\[
\varphi \ddfn   \vec p \subseteq \vec q,
\]
where $\vec p$ and $\vec q$ are finite tuples of proposition variables with the same length. 
The semantics for propositional inclusion atoms is defined as follows:
$$X\models {\vec p}\subseteq {\vec q}\text{ iff }\forall s\in X\, \exists t\in X :s(\vec p)=t(\vec q).$$

\begin{remark}
\emph{Extended propositional inclusion logic} is the variant of $\PLinc$ in which inclusion atoms of the form $\vec \varphi \subseteq \vec \psi$, where $\vec \varphi$ and $\vec \psi$ are tuples of $\PL$-formulae, are allowed. It is easy to see that this extension does not increase complexity of the logic and on that account, in this paper, we only consider the non-extended variant.
\end{remark}

It is easy to check that $\PLInc$  is not a downward closed logic\footnote{A logic \LL is downward closed if the implication $X\models \varphi$ and $Y\subseteq X$ $\Rightarrow$ $Y\models \varphi$ holds for every formula $\varphi\in \LL$ and teams $X$ and $Y$.}. However, analogously to FO-inclusion-logic \cite{Galliani12}, the same holds for $\PLInc$ w.r.t. unions: 

\begin{proposition}[Closure under unions]\label{closureunions}
Let $\varphi \in \PLInc$ and let $ X_i$, for $i\in I$, be teams. Suppose that $X_i\models \varphi$ for each $i\in I$. Then $\bigcup_{i\in I} X_i \models \varphi$. 
\end{proposition}

It is easy to see that, by Proposition \ref{whatever}, for propositional logic the strict and the lax semantics coincide; meaning that $X\models \varphi$ iff $X\models_s \varphi$ for all $X$ and $\varphi$. However this does not hold for propositional inclusion logic, for the following example shows that $\PLinc_s$ is not union closed. Moreover, we will show that the two different semantics lead to different complexities for the related model checking problems.
\begin{figure}[ht]
\begin{center}
\begin{minipage}[c]{0.38\textwidth}
\centering
\begin{tabular}{cC{3mm}C{3mm}C{3mm}}\toprule
&$p$& $q$ & $r$ \\\midrule
$s_1$&$1$ & $0 $ & $0$ \\
$s_2$&$1$ & $1 $ & $1$ \\
$s_3$&$0$ & $1 $ & $0$\\\bottomrule
\end{tabular}
\end{minipage}
\begin{minipage}[c]{0.38\textwidth}
\centering
\begin{tikzpicture}[dot/.style={circle,draw=black,fill=black,inner sep=0mm,minimum size=2mm}]
	\node at (-1,0) {$\mK:$};
	\node[dot,label={90:$w_1$}] (w1) at (0,0) {};
	\node[dot,label={90:$w_2$}] (w2) at (1,0) {};
	\node[dot,label={90:$w_3$}] (w3) at (2,0) {};

	\node[dot,label={-90:$s_1$}] (s1) at (0,-1.3) {};
	\node[dot,label={-90:$s_2$}] (s2) at (1,-1.3) {};
	\node[dot,label={-90:$s_3$}] (s3) at (2,-1.3) {};

\path[draw,black,-stealth'] 
(w1) edge node {} (s1)
(w2) edge node {} (s2)
(w3) edge node {} (s3)
(w1) edge node {} (s2)
(w3) edge node {} (s2);
\end{tikzpicture}
\end{minipage}
\end{center}

\caption{Assignments for teams in Example~\ref{ex:spincuc} and the Kripke model for Example \ref{ex:smincsing}.\label{pic0}}
\end{figure}
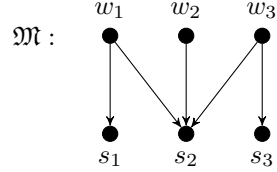

\begin{example}\label{ex:spincuc}
Let $s_1$, $s_2$, and $s_3$ be as in Table \ref{pic0} and define $\varphi:=\big(p\land (p\subseteq r)\big) \lor \big(q\land (q \subseteq r)\big)$. Note that $\{s_1,s_2\}\models_s \varphi$ and $\{s_2,s_3\}\models_s \varphi$, but $\{s_1,s_2,s_3\}\not\models_s\varphi$.
\end{example}

\section{Complexity of Propositional Inclusion Logic}

\begin{table*}\centering
\begin{tabular}{ccccccc}\toprule
	& \multicolumn{2}{c}{Satisfiability} & \multicolumn{2}{c}{Validity} & \multicolumn{2}{c}{Model checking}\\\cmidrule{2-3}\cmidrule{4-5}\cmidrule{6-7}
 & strict & lax  & strict & lax   & strict & lax \\\midrule
$\PL$ & \multicolumn{2}{c}{\Vhrulefill\; $\NP$ \cite{coo71a,levin73} \Vhrulefill}  & \multicolumn{2}{c}{\Vhrulefill\; $\co\NP$ \cite{coo71a,levin73} \Vhrulefill} & \multicolumn{2}{c}{\Vhrulefill\; $\NC{1}$ \cite{bus87} \Vhrulefill}\\
$\PLInc$ & $\EXPTIME$ \cite{hkmv15personal}
		    & $\EXPTIME$ \cite{hkmv15} & $\co\NP$ [Th.\ \ref{val-plinc-conp}]
		    & $\co\NP$ \cite{hkvv15} & $\NP$ [Th.\ \ref{mc-plinc-npc}]
		    & $\P$ [Th.\ \ref{mc-plinc-pc}]\\
\bottomrule\\
\end{tabular}
\caption{Complexity of the satisfiability, validity and model checking problems for  propositional logics under both systems of semantics. The shown complexity classes refer to completeness results.}
\label{pl-results}
\end{table*}

We now define the model checking, satisfiability, and validity problems in the context of team semantics.
Let $\LL$ be a  propositional logic with team semantics.  A formula $\varphi\in\LL$ is \emph{satisfiable}, if there exists a non-empty team $X$ such that $X\models \varphi$. A formula $\varphi\in\LL$ is \emph{valid}, if $X\models \varphi$ holds for all teams $X$ such that the proposition symbols in $\varphi$ are in the domain of $X$. The satisfiability problem $\SAT(\LL)$ and the validity problem $\VAL(\LL)$ are defined in the obvious way: Given a formula $\varphi\in\LL$, decide whether the formula is satisfiable (valid, respectively). For the model checking problem $\MC(\LL)$ we consider combined complexity: Given a formula $\varphi\in\LL$ and a team $X$, decide whether $X\models \varphi$. See Table \ref{pl-results} for known complexity results for $\PL$ and $\PLinc$, together with partial results of this paper.

It was shown in \cite{hkvv15} that the validity problem of $\PLInc$ is $\co\NP$-complete. 
Here we establish that the corresponding problem for $\PLInc_s$ is also $\co\NP$-complete. 
Our proof is similar to the one in \cite{hkvv15}. 
However the proof of  \cite{hkvv15} uses the fact that $\PLInc$ is union closed, while the same is not true for  $\PLInc_s$ (cf. Example \ref{ex:spincuc}).



\begin{lemma}\label{singletonslemma}
Let $X$ be a propositional team and $\varphi\in \PLinc_s$. If $\{s\}\models_s \varphi$ for every $s \in X$ then $X  \models_s \varphi$.
\end{lemma}\iflong
\begin{proof}
The proof is by a simple induction on the structure of the formula. 
The cases for atomic formulae and conjunction are trivial. 
The case for disjunction is easy: 
Assume that $\{s\}\models_s \varphi\lor \psi$ for every $s \in X$. 
Consequently for every $s \in X$ either $\{s\}\models_s \varphi$ or $\{s\}\models_s \psi$. 
As a result there exists $Y$ and $Z$ such that $Y\cup Z = X$, $Y\cap Z = \emptyset$, $\forall s\in Y: \{s\}\models_s \varphi$, and $\forall s\in Z: \{s\}\models_s \psi$. 
By the induction hypothesis  $Y\models_s \varphi$ and  $Z\models_s \psi$. 
Consequently, $X\models_s \varphi\lor \psi$.
\end{proof}
\fi

\begin{theorem}\label{val-plinc-conp}
The validity problem for $\PLInc_s$ is $\co\NP$-complete w.r.t. $\leqlogm$.
\end{theorem}\iflong
\begin{proof}
The $\co\NP$-hardness follows via Proposition \ref{whatever} from the fact that the validity problem of $\PL$ is  $\co\NP$-hard.  
Accordingly, it suffices to show $\VAL(\PLInc_s)\in \co\NP$. 
It is easy to check that, by Lemma \ref{singletonslemma}, a formula $\varphi\in \PLInc_s$ is valid iff it is satisfied by all singleton teams $\{ s\}$. 
Note also that, over  a singleton team $\{ s\}$, an inclusion atom
$ (p_1,\dots ,p_n) \subseteq (q_1,\dots ,q_n)$ is equivalent to the $\PL$-formula 
\[  \bigwedge_{1\le i\le n} (p_i\wedge q_i)\vee(\neg p_i\wedge\neg q_i). \]
Denote by $\varphi^*$ the $\PL$-formula obtained by replacing all inclusion atoms in $\varphi$ by their $\PL$-translations. 
By the above, $\varphi$ is valid iff $\varphi^*$ is valid. 
Since $\VAL(\PL)$ is in $\co\NP$ the claim follows. 
\end{proof}
\fi

\subsection{Model checking in lax semantics is P-complete}
In this section we construct a reduction from the monotone circuit value problem to the model checking problem of $\PLInc$.
For a deep introduction to circuits see \cite{vol99} by Vollmer.

\begin{definition}
A \emph{monotone Boolean circuit} with $n$ input gates and one output gate is a 3-tuple $C=(V,E,\alpha)$, where $(V,E)$ is a finite, simple, directed, acyclic graph, and $\alpha\colon V \rightarrow \{\lor,\land,x_1,\dots,x_n\}$ is a function such that the following conditions hold:
\begin{enumerate} 
\item{Every $v\in V$ has in-degree $0$ or $2$.}
\item {There exists exactly one $w\in V$ with out-degree $0$. We call this node $w$ the \emph{output gate} of $C$ and denote it by $g_\mathrm{out}$.}
\item{If $v \in V$ is a node with in-degree $0$, then $\alpha(v) \in \{x_1,\dots,x_n\}$}.
\item{If $v \in V$ has in-degree $2$, then $\alpha(v)\in\{\lor,\land\}$.} 
\item{For each $1 \leq i \leq n$, there exists exactly one $v \in V$ with $\alpha(v)=x_i$.} 
\end{enumerate} 

Let $C=(V,E,\alpha)$ be a monotone Boolean circuit with $n$ input gates and one output gate. 
Any sequence $b_1,\dots,b_n\in\{0,1\}$ of bits of length $n$ is called an \emph{input} to the circuit $C$. 
A function $\beta\colon V\rightarrow \{0,1\}$ defined such that
\[
\beta(v) \dfn
\begin{cases}
b_i &\text{if $\alpha(v)=x_i$}\\
\min\big(\beta(v_1),\beta(v_2)\big) & \text{if $\alpha(v)=\land$, where }v_1\neq v_2\text{ and }(v_1,v),(v_2,v)\in E, \\
\max\big(\beta(v_1),\beta(v_2)\big) & \text{if $\alpha(v)=\lor$, where }v_1\neq v_2\text{ and }(v_1,v),(v_2,v)\in E.
\end{cases}
\]
is called the \emph{valuation of the circuit $C$ under the input $b_1,\dots,b_n$}. The \emph{output of the circuit} $C$ is then defined to be $\beta(g_\mathrm{out})$.
\end{definition}

The \emph{monotone circuit value problem} ($\MCVP$) is the following decision problem: 
Given a monotone circuit $C$ and an input $b_1,\dots,b_n\in\{0,1\}$, is the output of the circuit $1$?

\begin{proposition}[\cite{mcvp77}]
$\MCVP$ is $\P$-complete w.r.t.\ $\leqlogm$ reductions.	
\end{proposition}

\begin{lemma}\label{Phard}
$\MC(\PLinc)$ under lax semantics is $\P$-hard w.r.t.\ $\leqlogm$.
\end{lemma}
\begin{proof}
We will establish a $\LOGSPACE$-reduction from $\MCVP$ to the model checking problem of $\PLinc$ under lax semantics. 
Since $\MCVP$ is $\P$-complete, the claim follows. 
More precisely, we will show how to construct, for each monotone Boolean circuit $C$ with $n$ input gates and for each input $\vec b$ for $C$, a team $X_{C,\vec b}$ and a $\PLinc$-formula $\varphi_C$ such that
$
X_{C,\vec b} \models \varphi_C  \text{ iff  the output of the circuit $C$ with the input $\vec b$ is $1$}. 
$

We use teams to encode valuations of the circuit. 
For each gate $v_i$ of a given circuit, we identify an assignment $s_i$. 
The crude idea is that if $s_i$ is in the team under consideration, the value of the gate $v_i$ with respect to the given input is $1$. 
The formula $\varphi_C$ is used to quantify a truth value for each Boolean gate of the circuit, and then for checking that the truth values of the gates propagate correctly.
We next define the construction formally and then discuss the background intuition in more detail.

Let $C=(V,E,\alpha)$ be a monotone Boolean circuit with $n$ input gates and one output gate and let $\vec{b}=(b_1\dots b_n) \in \{0,1\}^n$ be an input to the circuit $C$. 
We define that $V=\{v_0,\dots,v_m\}$ and that $v_0$ is the output gate of $C$.
Define
\begin{align*}
\tau_C \dfn \{p_0,\dots,p_m,p_\top,p_\bot\}\,\cup 
\{p_{k = i \lor j} \mid i < j, \alpha(v_k)=\lor, \text{ and }  (v_i,v_k),(v_j,v_k)\in E \}.	
\end{align*}

For each $i\leq m$, we define the assignment $s_i\colon\tau_C \to \{0,1\}$ as follows:
\[
s_i(p)\!\dfn\!
\begin{cases}
1 &\text{if $p=p_i$ or $p=p_\top$},\\
1 &\text{if $p=p_{k = i \lor j}$ or $p=p_{k = j \lor i}$ for some $j,k\leq m$},\\
0 &\text{otherwise}.
\end{cases}
\]

Furthermore, we define $s_\bot(p)=1$ iff $p=p_\bot$ or $p=p_\top$. 
We note that the assignment $s_\bot$ will be the only assignment that maps $p_{\bot}$ to 1. 
We make use of the fact that for each gate $v_i$ of $C$, it holds that $s_\bot(p_i)=0$. 
We define
\begin{align*}
X_{C,\vec b} \dfn \big\{s_i \mid \alpha(v_i)\in \{\land,\lor\}\big\}\,\cup
\big\{s_i \mid \alpha(v_i)\in \{x_i \mid b_i=1\}\big\}\cup\{s_\bot\},
\end{align*}
that is, $X_{C,\vec b}$ consists of assignments for each of the Boolean gates, assignments for those input gates that are given $1$ as an input, and of the auxiliary assignment $s_{\bot}$.

Let $X$ be any nonempty subteam of $X_{C,\vec b}$ such that $s_{\bot}\in X$. 
We have
\begin{align}
 X\models p_\top \subseteq p_0 &\text{ iff }&& s_0\in X \nonumber\\
   X\models p_i \subseteq p_j &\text{ iff }&& \text{($s_i\in X$ implies $s_j\in X$)}  \label{eq:idea}\\
   X\models p_k \subseteq p_{k=i\lor j} &\text{ iff }&& (i < j, (v_i,v_k),(v_j,v_k)\in E, \alpha(v_k)=\lor\nonumber\\
   &&&\text{ and $s_k\in X$ imply that $s_i\in X$ or $s_j\in X$)}\nonumber
\end{align}

Recall the intuition that $s_i\in X$ should hold iff the value of the gate $v_i$ is 1.
Define
\begin{align*} \label{eq1}
\psi_{\mathrm{out=1}} &\dfn p_\top \subseteq p_0,\\
\psi_{\land} &\dfn \bigwedge \{p_i \subseteq p_j \mid (v_j,v_i)\in E \text{ and } \alpha(p_i)=\land \},\\
\psi_{\lor} &\dfn \bigwedge \{p_k \subseteq p_{k=i\lor j} \mid i < j, (v_i,v_k)\in E, (v_j,v_k)\in E,\text{ and } \alpha(v_k)=\lor\},\\
\varphi_{C} &\dfn \neg p_\bot \lor (\psi_{\mathrm{out}=1} \land \psi_{\land} \land \psi_{\lor} ).
\end{align*}

It is quite straightforward to check 
\iflong (see details below)
\fi
that
$X_{C,\vec b} \models \varphi_C$ iff the output of $C$ with the input $\vec b$ is $1$.
\iflong

The idea of the reduction is the following: 
The disjunction in $\phi_C$ is used to guess a team $Y$ for the right disjunct that encodes the valuation $\beta$ of the circuit $C$. 
The right disjunct is then evaluated with respect to the team $Y$ with the intended meaning that $\beta(v_i)=1$ whenever $s_i\in Y$. 
Note that $Y$ is always as required in \eqref{eq:idea}. 
The formula $\psi_{\mathrm{out}=1}$ is used to state that $\beta(v_0)=1$, whereas the formulae $\psi_{\land}$ and $\psi_{\lor}$ are used to propagate the truth value $1$ down the circuit. 
The assignment $s_{\bot}$ and the proposition $p_\bot$ are used as an auxiliary to make sure that $Y$ is nonempty and to deal with the propagation of the value $0$ by the subformulae of the form $p_i \subseteq p_j$. 

Now observe that the team $X_{C,\vec b}$ can be easily computed by a logspace Turing machine which scans the input for $\land$-gates, $\lor$-gates, and true input gates, and then outputs the corresponding team members $s_i$ in a bitwise fashion. 
The formula $\varphi_{C}$ can be computed in logspace as well: 
\begin{enumerate}
	\item the left disjunct does not depend on the input,
	\item for $\psi_{\land}$ we only need to scan for the $\land$-gates and output the inclusion-formulae for the corresponding edges, 
	\item for $\psi_{\lor}$ we need to maintain two binary counters for $i$ and $j$, and use them for searching for those disjunction gates that satisfy $i<j$. 
\end{enumerate}
Consequently, the reduction can be computed in logspace. \fi
\end{proof}
For the proof of the above lemma it is not important that lax semantics is considered; the same proof works also for the strict semantics. However, as we will show next, we can show a stronger result for the model checking problem of $\PLinc_s$; namely that it is $\NP$-hard.
In Section \ref{sec:mcml} we will show that the model checking problem for modal inclusion logic with lax semantics is in $\P$ (Lemma \ref{Plemma}). 
Since $\PLinc$ is essentially a fragment of this logic, by combining Lemmas \ref{Phard} and \ref{Plemma}, we obtain the following theorem.
\begin{theorem}\label{mc-plinc-pc}
$\MC(\PLinc)$ under lax semantics is $\P$-complete w.r.t.\ $\leqlogm$.
\end{theorem}

\subsection{Model checking in strict semantics is NP-complete}
In this section we reduce the set splitting problem, a well-known $\NP$-complete problem, to the model checking problem of $\PLInc_s$.

\begin{definition}
The \emph{set splitting} problem is the following decision problem:
\begin{description}
	\item[\textit{Input:}] A family $\cF$ of subsets of a finite set $S$.
	\item[\textit{Problem:}] Do there exist subsets $S_1$ and $S_2$ of $S$ such that
\begin{enumerate}
\item $S_1$ and $S_2$ are a partition of $S$ (i.e., $S_1\cap S_2=\emptyset$ and $S_1\cup S_2= S$),
\item for each $A\in \cF$, there exist $a_1,a_2\in A$ such that $a_1\in S_1$ and $a_2\in S_2$?
\end{enumerate} 

\end{description}
\end{definition}

\begin{proposition}[\cite{gajo79}]
The set splitting problem is $\NP$-complete w.r.t.\ $\leqlogm$.
\end{proposition}

The following proof relies on the fact that strict semantics is considered. It cannot hold for lax semantics unless $\P=\NP$.
\begin{lemma}\label{mcisnph}
$\MC(\PLinc)$ under strict semantics is $\NP$-hard w.r.t.\ $\leqlogm$.
\end{lemma}
\begin{proof}
We give a reduction from the set splitting problem \cite[SP4]{gajo79} to the model checking problem of $\PLinc$ under strict semantics.

Let $\cF$ be an instance of the set splitting problem. 
We stipulate that $\cF=\{B_1,\dots,B_n\}$ and that $\bigcup \cF=\{a_1,\dots, a_k\}$, where $n,k\in\mathbb{N}$. 
We will introduce fresh proposition symbols $p_i$ and $q_j$ for each point $a_i\in \bigcup \cF$ and set $B_j\in \cF$. 
We will then encode the family of sets $\cF$ by assignments over these proposition symbols; each assignment $s_i$ will correspond to a unique point $a_i$. 
Formally, let $\tau_\cF$ denote the set $\{p_1,\dots, p_k, q_1,\dots,q_n,p_\top, p_c,p_d\}$ of proposition symbols. 
For each $i\in\{1,\dots,k,c,d\}$, we define the assignment $s_i\colon\tau_\cF \to \{0,1\}$ as follows:
\[
s_i(p) \dfn
\begin{cases}
1 &\text{if $p=p_i$ or $p=p_\top$},\\
1 &\text{if, for some $j$, $p=q_j$ and $a_i\in B_j$},\\
0 &\text{otherwise}.
\end{cases}
\]

Define $X_\cF\dfn \{s_1,\dots,s_k,s_c,s_d\}$, that is, $X_\cF$ consists of assignments $s_i$ corresponding to each of the points  $a_i\in \bigcup \cF$ and of two auxiliary assignments $s_c$ and $s_d$. 
Note that the only assignment in $X_\cF$ that maps $p_c$ ($p_d$, resp.) to 1 is $s_c$ ($s_d$, resp.) and that every assignment maps $p_\top$ to 1. 
Moreover, note that for $1\leq i\leq k$ and $1\leq j\leq n$, $s_i(q_j)=1$ iff $a_i\in B_j$. Now define
\[
\varphi_\cF \dfn  \big( \neg p_c\land \bigwedge_{i\leq n} p_\top \subseteq q_i \big) \lor \big( \neg p_d \land \bigwedge_{i\leq n} p_\top \subseteq q_i \big).
\]
We claim that
 $X_\cF\models_s \varphi_\cF$ iff the output of the set splitting problem with input $\cF$ is ``yes''.
\iflong

The proof is straightforward. 
Note that $X_\cF\models_s \varphi_\cF$ holds iff $X_\cF$ can be partitioned into two subteams $Y_1$ and $Y_2$ such that
\[
Y_1\models_s  \neg p_c\land \bigwedge_{i\leq n} p_\top \subseteq q_i  \text{ and }  Y_2\models_s  \neg p_d \land \bigwedge_{i\leq n} p_\top \subseteq q_i.
\]

Teams $Y_1$ and $Y_2$ are both nonempty, since $s_d\in Y_1$ and $s_c\in Y_2$. 
Also, for a nonempty subteam $Y$ of $X_\cF$, it holds that
$
Y\models_s p_\top \subseteq q_j$ iff there exists $s_i\in Y$ such that $s_i(q_j)=1$, or
equivalently, $a_i\in B_j$.

It is now evident that if $X_\cF\models_s \varphi_\cF$ holds then the related subteams $Y_1$ and $Y_2$ directly construct a positive answer to the set splitting problem. 
Likewise, any positive answer to the set splitting problem can be used to directly construct the related subteams $Y_1$ and $Y_2$.

In order to compute the assignments $s_i$ and by this the team $X_\cF$ on a logspace machine we need to implement two binary counters to count through $1\leq i\leq k$ for the propositions $p_i$ and $1\leq j\leq n$ for the propositions $q_i$. 
The formula $\varphi_\cF$ is constructed in logspace by simply outputting it step by step with the help of a binary counter for the interval $1\leq i\leq n$. 
As a result the whole reduction can be implemented on a logspace Turing machine.\fi
\end{proof}

In Section \ref{sec:mcml} we establish that the model checking problem of modal inclusion logic with strict semantics is in $\NP$ (Theorem \ref{mc-minc-npc}). 
Since $\PLinc$ is essentially a fragment of this logic, together with Lemma \ref{mcisnph}, we obtain the following theorem.
\begin{theorem}\label{mc-plinc-npc}
$\MC(\PLinc)$ under strict semantics is $\NP$-complete w.r.t.\ $\leqlogm$.
\end{theorem}

\section{Modal logics with team semantics}
Let $\Phi$ be a set of proposition symbols. The syntax of modal logic $\ML(\Phi)$ is  generated by the following grammar:
\[
\varphi \ddfn p\mid \neg p \mid (\varphi \wedge \varphi)\mid (\varphi \vee \varphi) \mid \Diamond \varphi \mid \Box \varphi, \text{where $p\in\Phi$.}
\]
By $\varphi^\bot$ we denote the formula that is obtained from $\neg \varphi$ by pushing all negation symbols to the atomic level.
A (Kripke) {\em $\Phi$-model} is a tuple $\mK=(W,R,V)$, where $W$, called the {\em domain} of $\mK$, is a non-empty set, $R\subseteq W\times W$ is a binary relation, and $V\colon\Phi\to \mathcal{P}(W)$ is a valuation of the proposition symbols.
By  $\modelsML$ we denote the \emph{satisfaction relation} of modal logic that is defined via pointed \emph{$\Phi$-models} in the standard way.
Any subset $T$ of the domain of a Kripke model $\mK$ is called \emph{a team of $\mK$}. 
Before we define \emph{team semantics} for $\ML$, we introduce some auxiliary notation.
\begin{definition}
Let $\mK=(W,R,V)$ be a model and $T$ and $S$ teams of $\mK$. Define that
\begin{center}
$R[T] := \{w\in W \mid \exists v\in T \text{ s.t. } vRw \}$ and
$R^{-1}[T] := \{w\in W \mid \exists v\in T \text{ s.t. }  wRv\}$.
\end{center}
For teams $T$ and $S$ of $\mK$, we write $T[R]S$ if $S\subseteq R[T]$ and $T\subseteq R^{-1}[S]$.
\end{definition}

Accordingly, $T[R]S$ holds if and only if for every $w\in T$, there exists some $v\in S$ such that $wRv$, and for every $v\in S$, there exists some $w\in T$ such that $wRv$.
We are now ready to define team semantics for $\ML$.
%
\begin{definition}[Lax team semantics]
Let $\mK$ be a Kripke model and $T$ a team of $\mK$. 
The satisfaction relation $\mK,T\models \varphi$ for $\ML(\Phi)$ is defined as follows. 
\begin{align*}
\mK,T\models p  \quad\Leftrightarrow\quad& w\in V(p) \, \text{ for every $w\in T$.}\\
\mK,T\models \neg p \quad\Leftrightarrow\quad& w\not\in V(p) \, \text{ for every $w\in T$.}\\
\mK,T\models (\varphi\land\psi) \quad\Leftrightarrow\quad& \mK,T\models\varphi \text{ and } \mK,T\models\psi.\\
\mK,T\models (\varphi\lor\psi) \quad\Leftrightarrow\quad& \mK,T_1\models\varphi \text{ and } 
\mK,T_2\models\psi \, \text{ for some $T_1$ and $T_2$ s.t.\ $T_1\cup T_2= T$}.\\
\mK,T\models \Diamond\varphi \quad\Leftrightarrow\quad& \mK,T'\models\varphi \text{ for some $T'$ s.t.\ $T[R]T'$}.\\
\mK,T\models \Box\varphi \quad\Leftrightarrow\quad& \mK,T'\models\varphi, \text{ where $T'=R[T]$}.
\end{align*}
\end{definition}

Analogously to the propositional case, we also consider the \emph{strict} variant of team semantics for modal logic.
In the \emph{strict} team semantics, we have the following alternative semantic definitions for the disjunction and diamond (where $W$ denotes the domain of \mK).
\begin{align*}
\mK,T\models_s (\varphi\lor\psi) \quad\Leftrightarrow\quad& \mK,T_1\models\varphi \text{ and } 
\mK,T_2\models\psi \,\\
&\text{for some $T_1$ and $T_2$ such that $T_1\cup T_2= T$ and $T_1\cap T_2= \emptyset$}.\\
\mK,T\models_s \Diamond\varphi \quad\Leftrightarrow\quad& \mK,f(T)\models\varphi \text{ for some $f\colon T\rightarrow W$ s.t.\ $\forall w\in T: wRf(w)$}.
\end{align*}

When $\LL$ is a team-based modal logic, we let $\LL_s$ to denote its variant with strict semantics. 
As in the propositional case, for strict semantics we use $\models_s$ instead of $\models$. 
The formulae of $\ML$ have the following flatness property.
\begin{proposition}[Flatness, see, e.g., \cite{dukovo16}]\label{flatness}
\label{mlextends}
Let $\mK$ be a Kripke model and $T$ be a team of $\mK$. Then, for every formula $\varphi$ of $\ML(\Phi)$:
\(
\mK,T\models \varphi \,\Leftrightarrow\, \forall w\in T:\mK,w\modelsML \varphi.
\)
\end{proposition}

The syntax of \emph{modal inclusion logic} $\Minc(\Phi)$ and \emph{extended modal inclusion logic} $\EMinc(\Phi)$ is obtained by extending the syntax of $\ML(\Phi)$ by the following grammar rule for each $n\in \N$:
$$
\varphi \ddfn {\varphi_1,\dots,\varphi_n \subseteq \psi_1,\dots,\psi_n},
$$
where $\varphi_1,\psi_1,\dots,\varphi_n,\psi_n\in\ML(\Phi)$. Additionally, for $\Minc(\Phi)$, we require that $\varphi_1$, $\psi_1$, $\dots$, $\varphi_n$, $\psi_n$ are proposition symbols in $\Phi$.
The semantics for these inclusion atoms is defined as follows:
\begin{align*}
&\mK,T\models \varphi_1,\dots,\varphi_n \subseteq \psi_1,\dots,\psi_n\Leftrightarrow \forall w\in T \exists v\in T: \bigwedge_{1 \leq i \leq n}(\mK,\{w\}\models\varphi_i \Leftrightarrow \mK,\{v\}\models\psi_i).
\end{align*}

The following proposition is proven in the same way as the analogous results for first-order inclusion logic \cite{Galliani12}.  A modal logic $\LL$ is union closed if $\mK,T\models \varphi$ and $\mK,S\models \varphi$ implies that  $\mK,T\cup S\models \varphi$, for every $\varphi\in \LL$.
\begin{proposition}[Union Closure]\label{closures}
The logics $\ML$, $\Minc$, $\EMinc$ are union closed.
\end{proposition}

\begin{table*}\centering
\begin{tabular}{@{}c@{\;}cc@{\;}cc@{\;}cc@{}}\toprule
	& \multicolumn{2}{c}{Satisfiability} & \multicolumn{2}{c}{Validity} & \multicolumn{2}{c}{Model checking}\\\cmidrule{2-3}\cmidrule{4-5}\cmidrule{6-7}
 & strict & lax  & strict & lax   & strict & lax \\\midrule
$\ML$ &  \multicolumn{2}{c}{\Vhrulefill\; $\PSPACE$ \cite{lad77} \Vhrulefill} & \multicolumn{2}{c}{\Vhrulefill\; $\PSPACE$ \cite{lad77} \Vhrulefill}  & \multicolumn{2}{c}{\Vhrulefill\; $\P$ \cite{clemsi86,sc02} \Vhrulefill}\\
%
%
$\Minc$   & $\EXPTIME$ \cite{hkmv15personal}
	      & $\EXPTIME$ \cite{hkmv15}
	      & $\coNEXPTIME$-h. [C.~\ref{cor:val-minc-strict}]
	      & $\coNEXPTIME$-h. [L.~\ref{lem:val-minc-lax-lower-bound}]
              & $\NP$ [Th.~\ref{mc-minc-npc}]
	      & $\P$ [Th.~\ref{mc-minc-pc}]
\\
\bottomrule\\
\end{tabular}
\caption{Complexity of satisfiability, validity and model checking for modal logics under both systems of semantics. The given complexity classes refer to completeness results and ``-{\normalfont h}.'' denotes hardness. 
The complexities for $\Minc$ and $\EMinc$ coincide,  see Theorems \ref{mc-minc-pc}, \ref{mc-minc-npc}, and \ref{a:val-eminc-minc-equiv}.}
\label{ml-results}
\end{table*}

Analogously to the propositional case, it is easy to see that, by Proposition \ref{flatness}, for $\ML$ the strict and the lax semantics coincide. 
Again, as in the propositional case, this does not hold for $\Minc$ or $\EMinc$. (Note that since $\PLinc_s$ is not union closed (cf. Example \ref{ex:spincuc}) neither $\Minc_s$ nor $\EMinc_s$ is as well.)

In contrary to the propositional case, Lemma \ref{singletonslemma} fails in the modal case as the following example illustrates.
\begin{example}\label{ex:smincsing}
Let $\mK$ be as depicted in the table of Figure~\ref{pic0} and let $\varphi$ denote the $\PLinc_s$-formula of Example \ref{ex:spincuc}. Now $\mK,\{w_i\}\models_s \Box \varphi$, for $i\in \{1,2,3\}$, but  $\mK,\{w_1,w_2,w_3\}\not\models_s \Box \varphi$.
\end{example}

\section{Model checking and validity in modal team semantics}
The model checking, satisfiability, and validity problems in the context of team semantics of modal logic are defined analogously to the propositional case.
Let $\LL(\Phi)$ be a  modal logic with team semantics.  
A formula $\varphi\in\LL(\Phi)$ is \emph{satisfiable}, if there exists a Kripke $\Phi$-model $\mK$ and a non-empty team $T$ of $\mK$ such that $\mK,T\models \varphi$. 
A formula $\varphi\in\LL(\Phi)$ is \emph{valid}, if $\mK,T\models \varphi$ holds for every $\Phi$-model $\mK$ and every team $T$ of $\mK$. 
The satisfiability problem $\SAT(\LL)$ and the validity problem $\VAL(\LL)$ are defined in the obvious way: 
Given a formula $\varphi\in\LL$, decide whether the formula is satisfiable (valid, respectively). 
For model checking  $\MC(\LL)$ we consider combined complexity: 
Given a formula $\varphi\in\LL$, a Kripke model $\mK$, and a team $T$ of $\mK$, decide whether $\mK,T\models \varphi$. 
See Table \ref{ml-results} for known complexity results on $\ML$ and $\Minc$, together with partial results of this paper.

\subsection{Complexity of model checking}\label{sec:mcml}

Let $\mK$ be a Kripke model, $T$ be a team of $\mK$, and $\varphi$ be a formula of $\MINC$. 
By $\maxsub(T,\varphi)$, we denote the maximum subteam $T'$ of $T$ such that $\mK,T'\models \varphi$. 
Since $\Minc$ is union closed (cf.  Proposition \ref{closures}), such a maximum subteam always exists.
\begin{lemma}\label{maxlemma}
If $\varphi$ is  a proposition symbol, its negation, or an inclusion atom, then $\maxsub(T,\varphi)$ can be computed in polynomial time with respect to $\lvert T \rvert + \lvert \varphi \rvert$.
\end{lemma}\iflong
\begin{proof} 
If $\varphi$ is a proposition symbol or its negation, the claim follows from flatness in a straightforward way.
Assume then that $T=\{w_1,\dots w_n\}$ and $\varphi=p_1,\dots,p_k \subseteq q_1,\dots,q_k$. Let $G=(V,E)$ be a directed graph such that $V=T$ and
$
(u,v)\in E$ iff the value of $p_i$ in $u$ is the same as the value of  $q_i$ in $v$, for each $1\leq i\leq k$.

The graph $G$ describes the inclusion dependencies between the points in the following sense: 
if $w \in \maxsub(T,\varphi)$, then there exists some $v\in \maxsub(T,\varphi)$ such that $(w,v)\in E$. 
Clearly $G$ can be computed in time  $\mathcal{O}(n^2k)$. 
In order to construct  $\maxsub(T,\varphi)$, we, round by round, delete all vertices from $G$ with out-degree $0$. 
Formally, we define a sequence $G_0,\dots,G_{n}$ of graphs recursively. 
We define that $G_0:=G$ and that $G_{j+1}$ is the graph obtained from $G_j$ by deleting all of those vertices from $G_j$ that have out-degree $0$ in $G_j$. 
Let $i$ be the smallest integer such that $G_i=(V_i,E_i)$ has no vertices of out-degree 0. 
Clearly $i\leq n$, and moreover, $G_i$ is computable from $G$ in time  $\mathcal{O}(n^3)$. 
It is easy to check that $V_i=\maxsub(T,\varphi)$.
\end{proof}\fi

For the following Lemma it is crucial that lax semantics is considered. The lemma cannot hold for strict semantics unless $\P=\NP$.
\begin{lemma}\label{Plemma}
$\MC(\Minc)$ under lax semantics is in $\P$.
\end{lemma}
\begin{proof}
We will present a labelling algorithm for model checking $\mK,T\models \varphi$. 
Let $\subOcc(\varphi)$ denote the set of all \emph{occurrences} of subformulae of $\varphi$. 
Below we denote occurrences as if they were formulae, but we actually refer to some particular occurrence of the formula.
  
A function $f\colon\subOcc(\varphi) \to \pow(W)$ is called a labelling function of $\varphi$ in $\mK$.
We will next give an algorithm for computing a sequence $f_0,f_1,f_2,\ldots$, of 
such labelling functions.
\begin{itemize}
\item Define $\lf_0(\psi)= W$ for each $\psi\in \subOcc(\varphi)$.

\item For odd $i\in\mathbb{N}$, define $f_i(\psi)$ bottom up as follows:
\begin{enumerate}
\item For literal $\psi$, define $\lf_{i}(\psi) \dfn \maxsub(\lf_{i-1}(\psi),\psi).$
\item $\lf_{i}(\psi\land\theta) \dfn \lf_i(\psi)\cap\lf_i(\theta)$.
\item $\lf_{i}(\psi\lor\theta) \dfn \lf_i(\psi)\cup\lf_i(\theta)$.
\item $\lf_{i}(\Diamond\psi) \dfn \{w\in \lf_{i-1}(\Diamond\psi) \mid R[w]\cap \lf_{i}(\psi)\neq \emptyset\}$.
\item $\lf_{i}(\Box\psi) \dfn \{w\in \lf_{i-1}(\Box\psi) \mid R[w]\subseteq \lf_{i}(\psi)\}$.
\end{enumerate}

\item For even $i\in\mathbb{N}$ larger than $0$, define $f_i(\psi)$ top to bottom as follows:
\begin{enumerate}
\item Define $\lf_{i}(\varphi) \dfn \lf_{i-1}(\varphi) \cap T$.
\item If $\psi=\theta \land \gamma$, define $\lf_i(\theta) \dfn \lf_i(\gamma) \dfn \lf_{i}(\theta\land\gamma)$.
\item If $\psi=\theta \lor \gamma$, define $\lf_i(\theta) \dfn  \lf_{i-1}(\theta)\cap\lf_{i}(\theta\lor\gamma)$ and $\lf_i(\gamma) \dfn  \lf_{i-1}(\gamma)\cap\lf_{i}(\theta\lor\gamma)$.
\item If $\psi=\Diamond\theta$, define $\lf_i(\theta) \dfn \lf_{i-1}(\theta) \cap R[\lf_i(\Diamond\theta)]$.
\item If $\psi=\Box\theta$, define $\lf_i(\theta) \dfn \lf_{i-1}(\theta) \cap R[\lf_i(\Box\theta)]$.
\end{enumerate}
\end{itemize}
By a straightforward induction on $i$, we can prove that $f_{i+1}(\psi)\subseteq f_i(\psi)$ holds for every $\psi\in\subOcc(\varphi)$. 
The only nontrivial induction step is that for $f_{i+1}(\theta)$ and $f_{i+1}(\gamma)$, when $i+1$ is even and $\psi=\theta \land \gamma$. 
To deal with this step, observe that, by the definition of $f_{i+1}$ and $f_i$, we have 
$f_{i+1}(\theta)=f_{i+1}(\gamma)=f_{i+1}(\psi)$ and $f_i(\psi)\subseteq f_i(\theta),f_i(\gamma)$,
and by the induction hypothesis on $\psi$, we have $f_{i+1}(\psi)\subseteq f_i(\psi)$.

It follows that there is an integer $j\le 2\cdot|W|\cdot |\varphi|$ such that $f_{j+2}=f_{j+1}=f_j$.
We denote this fixed point $f_j$ of the sequence $f_0,f_1,f_2,\ldots$ by $f_\infty$. 
By Lemma \ref{maxlemma} the outcome of $\maxsub(\cdot,\cdot)$ is computable in polynomial time with respect to its input. 
That being, clearly $f_{i+1}$ can be computed from $f_i$ in polynomial time with respect to $|W|+|\varphi|$. 
On that account $f_\infty$ is also computable in polynomial time with respect to $|W|+|\varphi|$.

We will next prove by induction on $\psi\in\subOcc(\varphi)$ that $\mK,f_\infty(\psi)\models\psi$. 
Note first that there is an odd integer $i$ and an even integer $j$ such that $f_\infty=f_i=f_j$.
\begin{enumerate}
\item If $\psi$ is a literal, the claim is true since $f_\infty=f_i$ and  
$\lf_{i}(\psi)=\maxsub(\lf_{i-1}(\psi),\psi)$. 
\item Assume next that $\psi=\theta\land\gamma$, and the claim holds for $\theta$ and $\gamma$. 
	Since $f_\infty=f_j$, we have $f_\infty(\psi)=
f_\infty(\theta)=f_\infty(\gamma)$, as a result, by induction hypothesis, $\mK,f_\infty(\psi)\models\theta\land\gamma$, as desired.
\item In the case $\psi=\theta\lor\gamma$, we obtain the claim $\mK,f_\infty(\psi)\models\psi$ by using the induction hypothesis, and the observation that 
$f_\infty(\psi)=f_i(\psi)=f_i(\theta)\cup f_i(\gamma)=f_\infty(\theta)\cup f_\infty(\gamma)$.
\item Assume then that $\psi=\Diamond\theta$. Since $f_\infty=f_i$, we have
$f_\infty(\psi) = \{w\in f_{i-1}(\psi) \mid R[w]\cap f_\infty(\theta)\neq \emptyset\}$, as a consequence $f_\infty(\psi)\subseteq R^{-1}[f_\infty(\theta)]$. 
On the other hand, since $f_\infty=f_j$, we have $f_\infty(\theta) = f_{j-1}(\theta) \cap R[f_\infty(\psi)]$, for this reason
$f_\infty(\theta)\subseteq R[f_\infty(\psi)]$. 
Accordingly, $f_\infty(\psi)[R]f_\infty(\theta)$, and using the induction hypothesis, we see that $\mK,f_\infty(\psi)\models\psi$.
\item Assume finally that $\psi=\Box\theta$.  
Since $f_\infty=f_i$, we have $R[f_\infty(\psi)]\subseteq f_\infty(\theta)$. 
On the other hand, since $f_\infty=f_j$, we have $f_\infty(\theta)\subseteq R[f_\infty(\psi)]$. 
This shows that $f_\infty(\theta)=R[f_\infty(\psi)]$, that being the case by the induction hypothesis, $\mK,f_\infty(\psi)\models\psi$.
\end{enumerate}

In particular, if $f_\infty(\varphi)=T$, then $\mK,T\models\varphi$. 
Consequently, to complete the proof of the lemma, it suffices to prove that the converse implication is true, as well. 
To prove this, assume that $\mK,T\models\varphi$.
Then for each $\psi\in\subOcc(\varphi)$, there is a team $T_\psi$ such that 
\begin{enumerate}
\item $T_\varphi=T$.
\item If $\psi=\theta \land \gamma$, then $T_\psi=T_\theta=T_\gamma$.
\item If $\psi=\theta \lor \gamma$, then $T_\psi=T_\theta\cup T_\gamma$.
\item If $\psi=\Diamond\theta$, then $T_\psi[R] T_\theta$.
\item If $\psi=\Box\theta$, then $T_\theta=R[T_\psi]$.
\item If $\psi$ is a literal, then $\mK,T_\psi\models\psi$.
\end{enumerate}
We prove by induction on $i$ that $T_\psi\subseteq f_i(\psi)$ for  all
$\psi\in\subOcc(\varphi)$. For $i=0$, this is obvious, since $f_0(\psi)=W$ for all $\psi$. 
Assume next that $i+1$ is odd and the claim is true for $i$.
We prove the claim $T_\psi\subseteq f_i(\psi)$ by induction on $\psi$.
\begin{enumerate}
\item If $\psi$ is a literal, then $f_{i+1}(\psi)=\maxsub(f_i(\psi),\psi)$. Since
$\mK, T_\psi\models\psi$, and by induction hypothesis, $T_\psi\subseteq f_i(\psi)$, the claim $T_\psi\subseteq f_{i+1}(\psi)$ is true.
\item Assume that $\psi=\theta\land\gamma$. 
By induction hypothesis on $\theta$ and $\gamma$, we have $T_\psi=T_\theta\subseteq f_{i+1}(\theta)$ and $T_\psi=T_\gamma\subseteq f_{i+1}(\gamma)$. 
For this reason, we get $T_\psi\subseteq f_{i+1}(\theta)\cap f_{i+1}(\gamma)=f_{i+1}(\psi)$.
\item The case $\psi=\theta\lor\gamma$ is similar to the previous one; we omit the details.
\item If $\psi=\Diamond\theta$, then 
$f_{i+1}(\psi) = \{w\in f_i(\psi) \mid R[w]\cap f_{i+1}(\theta)\neq \emptyset\}$.
By the two induction hypotheses on $i$ and $\theta$, we have
$\{w\in T_\psi \mid R[w]\cap T_\theta\neq \emptyset\}\subseteq f_{i+1}(\psi)$.
The claim follows from this, since the condition $R[w]\cap T_\theta\neq \emptyset$ holds for all $w\in T_\psi$.
\item The case $\psi=\Box\theta$ is again similar to the previous one, so  we omit the details.
\end{enumerate}
Assume then that $i+1$ is even and the claim is true for $i$. This time
we prove the claim $T_\psi\subseteq f_i(\psi)$ by top to bottom induction on $\psi$.
\begin{enumerate}
\item By assumption, $T_\varphi=T$, on that account by induction hypothesis, 
$T_\varphi\subseteq f_i(\varphi)\cap T=f_{i+1}(\varphi)$.
\item Assume that $\psi=\theta\land\gamma$. By induction hypothesis on $\psi$,  we have $T_\psi\subseteq f_{i+1}(\psi)$. 
Since $T_\psi=T_\theta=T_\gamma$ and $f_{i+1}(\psi)=f_{i+1}(\theta)=f_{i+1}(\gamma)$, this implies that $T_\theta\subseteq f_{i+1}(\theta)$ and $T_\gamma\subseteq f_{i+1}(\gamma)$.
\item Assume that $\psi=\theta\lor\gamma$. Using the fact that $T_\theta\subseteq T_\psi$, and the two induction hypotheses on $i$ and $\psi$, we see that $T_\theta\subseteq f_i(\theta)\cap T_\psi \subseteq f_i(\theta)\cap f_{i+1}(\psi)= f_{i+1}(\theta)$.
Similarly, we see that $T_\gamma\subseteq f_{i+1}(\gamma)$.
\item Assume that $\psi=\Diamond\theta$.
By the induction hypothesis on $i$, we have 
$T_\theta\subseteq f_i(\theta)$, and by the induction hypothesis on $\psi$, we have 
$T_\theta\subseteq R[T_\psi]\subseteq R[f_{i+1}(\psi)]$. Accordingly, we see that 
$T_\theta\subseteq f_i(\theta) \cap R[f_{i+1}(\psi)]=f_{i+1}(\theta)$.
\item The case $\psi=\Box\theta$ is similar to the previous one;
we omit the details.
\end{enumerate}
It follows now that $T=T_\varphi\subseteq f_\infty(\varphi)$. Since $f_\infty(\varphi)
\subseteq f_2(\varphi)\subseteq T$, we conclude that $f_\infty(\varphi)=T$.
This completes the proof of the implication $\mK,T\models\varphi\;\Rightarrow\;
f_\infty(\varphi)=T$.
\end{proof}
\begin{lemma}\label{Plemma2}
$\MC(\EMinc)$ under lax semantics is in $\P$.
\end{lemma}\iflong
\begin{proof}
The result follows by a polynomial time reduction to the model checking problem of $\Minc$: 
Let $(W,R,V),T$ be a team pointed Kripke model and $\varphi$ be a formula of $\EMinc$. 
Let $\varphi_1,\dots,\varphi_n$ be exactly those subformulae of $\varphi$ that occur as a parameter of some inclusion atom in $\varphi$ and let $p_1,\dots,p_n$ be distinct fresh proposition symbols. 
Let $V'$ be a valuation defined as follows
\[
V'(p):=
\begin{cases}
\{w \in W\mid (W,R,V),w \modelsML \varphi_i \} & \text{ if $p=p_i$},\\
V(p) &\text{ otherwise}.
\end{cases}
\]
Let $\varphi^*$ denote the formula obtained from $\varphi$ by simultaneously substituting each $\varphi_i$ by $p_i$. 
It is easy to check that $(W,R,V),T\models \varphi$ if and only if $(W,R,V'),T\models \varphi^*$. 
Moreover, $\varphi^*$ can be clearly computed form $\varphi$ in polynomial time. 
Likewise, $V'$ can be computed in polynomial time; since each $\varphi_i$ is a modal formula the truth set of that formula in $(W,R,V)$ can be computed in polynomial time by the standard labelling algorithm used in modal logic (see e.g., \cite{blackburn}), and the numbers of such computations is bounded above by the size of $\varphi$. 
As a consequence the result follows form Lemma \ref{Plemma}.
\end{proof}\fi

By combining Lemmas \ref{Phard}, \ref{Plemma}, and \ref{Plemma2}, we obtain the following theorem.
\begin{theorem}\label{mc-minc-pc}
$\MC(\Minc)$ and $\MC(\EMinc)$ under lax semantics are $\P$-complete w.r.t. $\leqlogm$.
\end{theorem}


\begin{theorem}\label{mc-minc-npc}
$\MC(\Minc)$ and $\MC(\EMinc)$ under strict semantics are $\NP$-complete w.r.t. $\leqlogm$.
\end{theorem}\iflong
\begin{proof}
The  $\NP$-hardness follows from the propositional case, i.e., by Lemma \ref{mcisnph}. 

The obvious brute force algorithm for model checking for $\EMinc$ works in $\NP$: 
For disjunctions and diamonds, we use nondeterminism to guess the correct partitions or successor teams, respectively. 
Conjunctions are dealt sequentially and for boxes the unique successor team can be computed by brute force in quadratic time. 
Checking whether a team satisfies an inclusion atom or a (negated) proposition symbol can be computed by brute force in polynomial time (this also follows directly from Lemma \ref{maxlemma}).
\end{proof}\fi

\iflong
\subsection{Dependency quantifier Boolean formulae}
Deciding whether a given quantified Boolean formula (qBf) is valid is a canonical $\PSPACE$-complete problem. \emph{Dependency quantifier Boolean formulae} introduced by Peterson et~al. \cite{Peterson2001} are variants of qBfs for which the corresponding decision problem is $\NEXPTIME$-complete. 
In this section, we define the related $\coNEXPTIME$-complete complementary problem. 
For the definitions related to dependency quantifier Boolean formulae, we follow  Virtema \cite{virtema14}.
 
QBfs extend propositional logic by allowing a prenex quantification of proposition symbols.
Formally, the set of \emph{qBfs} is built from formulae of propositional logic by the following grammar:
$$
\varphi \ddfn \exists p\, \varphi \mid \forall p\, \varphi \mid \theta,
$$
where $p$ is a propositional variable (i.e., a proposition symbol) and $\theta$ is formula of propositional logic.
The semantics for qBfs is defined via assignments $s\colon\mathrm{PROP}\to \{0,1\}$ in the obvious way.
When $C$ is a set of propositional variables, we denote by $\vec{c}$ the canonically ordered tuple of the variables in the set $C$.
When $p$ is a propositional variable and $b\in \{0,1\}$ is a truth value, we denote by $s(p\mapsto b)$ the modified assignment defined as follows:
\[
s(p\mapsto b)(q)\dfn
\begin{cases}
b & \text{if $q = p$},\\
s(q) &\text{otherwise}.
\end{cases}
\]

A formula that does not have any free variables is called \emph{closed}. 
We denote by $\QBF$ the set of exactly all closed quantified Boolean formulae.

\begin{proposition}[\cite{Stockmeyer:1973}]\label{QBFPSPACE}
The validity problem of $\QBF$ is $\PSPACE$-complete w.r.t.\ $\leqlogm$.
\end{proposition}

A \emph{simple qBf} is a closed qBf of the type
$
\varphi \dfn \forall p_1 \cdots \forall p_{n}\exists q_1\cdots \exists q_k\theta,
$
where $\theta$ is a propositional formula and the propositional variables $p_i,q_j$ are all distinct. 
Any tuple $(C_1,\dots,C_k)$ such that $C_1,\dots,C_k\subseteq \{p_1,\dots,p_n\}$ is called a \emph{constraint} for $\varphi$.
%
%
%
Intuitively, a constraint $C_j = \{ p_1, p_3 \}$ can be seen as a dependence atom $\dep{p_1,p_3, q_j}$\footnote{See Section \ref{sec:further} for a definition.}. 
A constraint $C_j = \{ p_1, p_3 \}$ can be also interpreted to indicate that the semantics of $\exists g_j$ is defined, if skolemised, via a Skolem function $f_j(p_1,p_3)$.
\begin{definition}
A simple qBf $\forall p_1 \cdots \forall p_{n}\exists q_1\cdots \exists q_k\theta$ is \emph{valid under a constraint} $(C_1,\dots,C_k)$, if there exist functions $f_1,\dots,f_k$ with
$
f_i\colon \{0,1\}^{\lvert C_i\rvert}\to \{0,1\}
$
such that for each assignment $s\colon \{p_1,\dots,p_n\}\to\{0,1\}$,
$
s(q_1 \mapsto f_1(s(\vec{c}_1)), \dots, q_k \mapsto f_k(s(\vec{c}_k)))\models \theta.
$
\end{definition}

A \emph{dependency quantifier Boolean formula} is a pair $(\varphi, \vec C)$, where $\varphi$ is a simple quantified Boolean formula and $\vec{C}$ is a constraint for $\varphi$. 
We say that $(\varphi, \vec C)$ is \emph{valid}, if $\varphi$ is valid under the constraint $\vec C$.
Let $\DQBF$ denote the set of all dependency quantifier Boolean formulae.
\begin{proposition}[{\cite[5.2.2]{Peterson2001}}]\label{TDQBF problem}
The validity problem of $\DQBF$ is $\NEXPTIME$-complete w.r.t.\ $\leqlogm$.
\end{proposition}
\begin{definition}\label{def:nonvalid}
Given a simple qBf $\forall p_1 \cdots \forall p_{n}\exists q_1\cdots \exists q_k\theta$, we say it is \emph{non-valid under a constraint} $(C_1,\dots,C_k)$, if for all functions $f_1,\dots,f_k$ with
$
f_i\colon \{0,1\}^{\lvert C_i\rvert}\to \{0,1\}, 
$
there exists an assignment $s\colon \{p_1,\dots,p_n\}\to\{0,1\}$ such that
$s(q_1 \mapsto f_1(s(\vec{c}_1)), \dots, q_k \mapsto f_k(s(\vec{c}_k)))\not\models \theta.
$
\end{definition}

It is straightforward to see that non-validity problem of $\DQBF$ is the complement problem of the validity problem of $\DQBF$. 
Accordingly, the following corollary follows.
\begin{corollary}\label{nonTDQBF problem}
The non-validity problem of $\DQBF$ is $\coNEXPTIME$-complete w.r.t.\ $\leqlogm$.
\end{corollary}

\subsection{Complexity of the validity problem is coNEXP-hard}\label{sec:val}
In this  section we give a reduction from the non-validity problem of $\DQBF$ to the validity problem of $\Minc$. 
\fi
\begin{lemma}\label{lem:val-minc-lax-lower-bound} 
$\VAL(\Minc)$ under lax semantics is $\coNEXPTIME$-hard w.r.t.\ $\leqlogm$.
\end{lemma}\iflong
\begin{proof}
We provide a $\leqlogm$-reduction from the non-validity problem of $\DQBF$
to the validity problem of $\Minc$. 

Recall Definition \ref{def:nonvalid}. 
In our reduction we will encode all the possible modified assignments of Definition \ref{def:nonvalid} by points in Kripke models.
First we enforce binary (assignment) trees of depth $n$ in our structures. 
Leafs of the binary tree will correspond to the set of assignments $s\colon \{p_1,\dots,p_n\}\to\{0,1\}$.
The binary trees are forced in the standard way by modal formulae:
The formula $\branch{p_i}\dfn \Diamond p_i\land\Diamond\lnot p_i$ forces that there are $\ge 2$ successor states which disagree on a proposition $p_i$.
The formula $\store{p_i}\dfn (p_i\to\Box p_i)\land(\lnot p_i\to\Box\lnot p_i)$ is used to propagate chosen values for $p_i$ to successors in the tree.
Now define
\begin{align*}
 \tree{p,n}\dfn \branch{p_1}\land\bigwedge_{i=1}^{n-1}\Box^i\Bigl(\branch{p_{i+1}}\land\bigwedge_{j=1}^i\store{p_j}\Bigr),
\end{align*}
where $\Box^i\varphi\dfn\overbrace{\Box\cdots\Box}^{i\text{ many}}\varphi$ is the $i$-times concatenation of $\Box$. 
The formula  $\tree{p,n}$ forces a complete binary assignment tree of depth $n$ for proposition symbols $p_1,\dots,p_n$. 
Notice that $\tree{p,n}$ is an $\ML$-formula and consequently flat (see Proposition \ref{flatness}). Let $\ell:=\mathrm{max}\{ \lvert C_1\rvert,\dots,\lvert C_k\rvert \}.$
Then define
\begin{align*}
\varphi_{\text{struc}}\dfn
&\tree{p,n}\land
\,\Box^n\bigl(\tree{t,\ell}\bigr)\land
\Box^{n+\ell}\bigl((p_\theta\leftrightarrow\theta) \land p_\top \land \neg p_\bot \bigr)\\
&\land\Box^n\Bigl(\;\bigwedge_{\mathclap{1\leq i\leq \ell}}\Box^i\bigl(\;\bigwedge_{\mathclap{1\leq j\leq n}}\store{p_j} \land \bigwedge_{\mathclap{1\leq r\leq k}}\store{q_r}\bigr) \Bigr).
\end{align*}

The formula $\varphi_{\text{struc}}$ enforces the full binary assignment tree w.r.t.\ the $p_i$s, enforces in its leaves trees of depth $\ell$ for variables $t_i$, identifies the truth of $\theta$ by a proposition $p_\theta$ at the depth $n+\ell$ as well as $1$ by $t_\top$ and $0$ by $t_\bot$, and then stores the values of the $p_j$s and $q_r$s consistently in their subtrees of relevant depth. 
The points at depth $n$ are used to encode the modified assignments of Definition \ref{def:nonvalid}.

Recall again Definition \ref{def:nonvalid} and consider the simple qBf $\forall p_1 \cdots \forall p_{n}\exists q_1\cdots \exists q_k\theta$ with constraint $(C_1,\dots,C_k)$. 
Then consider some particular Kripke model with the structural properties described above. 
Shift your attention to those points in the enforced tree that are in depth $n$. 
Note first that if we restrict our attention to proposition symbols $p_1, \dots, p_{n}$ all assignments are present. 
In fact the number points corresponding to some particular assignment can by anything $\geq1$.
Values for the proposition symbols $q_j$ and consequently for the functions $f_j$ arise from the particular model; essentially, since we are considering validity, all possible values will be considered. 
In fact, in some particular models, the values of $q_j$s are not functionally determined according to the related constraint $C_j$.
We will next define a formula that will deal with those models in which, for some $j$, the values for $q_j$ do not give rise to a function $f_j$ in Definition \ref{def:nonvalid}.
%
%
%
%
These unwanted models have to be ``filtered'' out by the formula through satisfaction.
This violation is expressed via $\varphi_{\text{cons}}$ defined as follows. Below we let $n_j = \lvert C_j \rvert$.
\begin{align}\label{quantification}
	\varphi_{\text{cons}}\dfn\bigvee_{\mathclap{\substack{1\leq j\leq k,\\C_j=\{\,p_{i_1},\dots,p_{i_{ n_j }}\}}}}\;\;(t_1\cdots t_{n_j} t_\bot\subseteq p_{i_1}\cdots p_{i_{n_j}}q_j)\land(t_1\cdots t_{n_j} t_\top\subseteq p_{i_1}\cdots p_{i_{n_j}}q_j).
\end{align}
Assume that $t_\top$ and $t_\bot$ correspond to the constant values $1$ and $0$, respectively.
Moreover, for the time being, suppose that the values for the proposition symbols $t_i$ have been existentially quantified (we will later show how this is technically done).
Now the formula $\varphi_{\text{cons}}$ essentially states that there exists a $q_j$ that does not respect the constraint $C_j$.

Finally define
\begin{align}\label{quantification2}
\varphiNonVal\dfn \varphi_{\text{struc}}^\bot\lor \Bigl (\varphi_{\text{struc}} \land
\Box^n\bigl(\Diamond^\ell(\varphi_{\text{cons}}\lor p_\bot\subseteq p_\theta\bigr)\Bigr).
\end{align}

By $\varphi_{\text{struc}}^\bot$, we denote the negation normal form of the $\ML$-formula $\neg \varphi_{\text{struc}}$.  
An important observation is that since $\varphi_{\text{struc}}$ is an $\ML$-formula, it is flat. 
Now the formula $\varphiNonVal$ is valid if and only if
\begin{align}\label{consistency}
\mK,T\models \Box^n\bigl(\Diamond^\ell(\varphi_{\text{cons}}\lor p_\bot\subseteq p_\theta)\bigr)
\end{align}
 holds for every team pointed Kripke model $\mK,T$ that satisfies the structural properties forced by $\varphi_{\text{struc}}$. 
 Let us now return to the formula \eqref{quantification}. 
 There, we assumed that the proposition symbols $t_i$ had been quantified and that the symbols $p_\top$ and $p_\bot$ correspond to the logical constants. 
 The latter part we already dealt with in the formula $\varphi_{\text{struc}}$. 
 Recall that $\varphi_{\text{struc}}$ forces full binary assignment trees for the $t_i$s that start from depth $n$. 
 The quantification of the $t_i$s is done by selecting the corresponding successors by the diamonds $\Diamond^\ell$ in the formula \eqref{quantification2}.
If $\mK,T$ is such that, for some $j$, $q_j$ does not respect the constraint $C_j$, we use $\Diamond^\ell$ to guess a witness for the violation. 
It is then easy to check that the whole team obtained by evaluating the diamond prefix satisfies the formula $\varphi_{\text{cons}}$.
On the other hand, if $\mK,T$ is such that for each $j$ the value of $q_j$ respects the constraint $C_j$, then the subformula $p_\bot\subseteq p_\theta$ forces that there exists a point $w$ in the team obtained from $T$ by evaluating the modalities in \eqref{consistency} such that $\mK,\{w\}\not\models p_\theta$. 
In our reduction this means that $w$ gives rise to a propositional assignment that falsifies $\theta$ as required in Definition \ref{def:nonvalid}.

It is now quite straightforward to show that a simple qBf $\forall p_1 \cdots \forall p_{n}\exists q_1\cdots \exists q_k\theta$ is \emph{non-valid under a constraint} $(C_1,\dots,C_k)$ iff the $\Minc$-formula $\varphiNonVal$ obtained as described above is valid. 

In the following we show the correctness of the constructed reduction. By the observation made in \eqref{consistency} it suffices to show the following claim:
\begin{claim}
	$\forall p_1 \cdots \forall p_{n}\exists q_1\cdots \exists q_k\theta$ is non-valid under $(C_1,\dots,C_k)$ iff $\mK,T\models \Box^n\bigl(\Diamond^\ell(\varphi_{\text{cons}}\lor p_\bot\subseteq p_\theta)\bigr)$ 
 holds for every team pointed Kripke model $\mK,T$ that satisfies the structural properties forced by $\varphi_{\text{struc}}$.
\end{claim}


\begin{proof}[Proof of Claim]
	``$\Rightarrow$'': Assume that the formula $\varphi\dfn\forall p_1 \cdots \forall p_{n}\exists q_1\cdots \exists q_k\theta$ is non-valid under the constraint $(C_1,\dots,C_k)$. 
	As a consequence, for every sequence of functions $f_1,\dots,f_k$ of appropriate arities there exists an assignment $s\colon\{p_1,\dots,p_n\}\to\{0,1\}$ such that
	\begin{equation}\label{eq:modified}
	s(q_1 \mapsto f_1\big(s(\vec{c}_1)), \dots, q_k \mapsto f_k(s(\vec{c}_k)))\not\models \theta.
	\end{equation}
	We will show that
	\begin{equation}\label{eq:boxed}
	\mK,T\models \Box^n\bigl(\Diamond^\ell(\varphi_{\text{cons}}\lor p_\bot\subseteq p_\theta)\bigr),
	\end{equation}
	for each team pointed Kripke model $\mK,T$ that satisfies the structural properties forced by $\varphi_{\text{struc}}$. 
	
    Let $\mK,T$ be an arbitrary team pointed Kripke model that satisfies the required structural properties. 
    Denote by $S$ the team obtained from $T$ after evaluating the first $n$ $\Box$-symbols in \eqref{eq:boxed}. 
    Note that each tuple of values assigned to $\vec p\dfn (p_1,\dots,p_n)$ is realised in $S$ as the tree structure enforces all possible assignments over $\vec p$. 
    Due to the forced structural properties, $S$ and of any team obtainable from $S$ by evaluating the $k$ $\Diamond$-symbols in \eqref{eq:boxed} realise exactly the same assignments for $\{p_1,\dots,p_n,q_1,\dots q_k\}$. 
    Let $S_k$ denote the set of exactly all points reachable from $S$ by paths of length exactly $k$. 
    For each point $w$ denote by $w(\vec q)$ the value of $\vec q$ in the world $w$.  
    Note that for every $\ell$-tuple of bits $\vec{b}$ and every point $w\in S$ there exists a point $v\in S_k$ such that $v(p_1,\dots,p_n,q_1,\dots q_k)=w(p_1,\dots,p_n,q_1,\dots q_k)$ and $v(t_1,\dots,t_\ell)=\vec{b}$. 
    Moreover, for any fixed $\vec{b}$, the team
    \[
    \{w\in S_k \mid w(t_1,\dots,t_\ell)= \vec{b}\}
    \]
is obtainable from $S$ by evaluating the $k$ $\Diamond$-symbols in \eqref{eq:boxed}.
We have two cases:		
	\begin{enumerate}
		\item There exists a constraint $C_i$, $1\leq i\leq k$, and points $w,w'\in S$ with $w(\vec{c_i})=w'(\vec{c_i})$ but $w(q_i)\neq w'(q_i)$.
		Now let $S'$ be a team obtained from $S$  by evaluating the $k$ $\Diamond$-symbols in \eqref{eq:boxed} such that, for every $w'\in S'$, $w(t_1,\dots,t_\ell)$ is an expansion of $w(\vec{c_i})$. 
		Now clearly $\mK,S'\models \varphi_{\text{cons}}$ and as a consequence $\mK,S\models \Diamond^\ell(\varphi_{\text{cons}}\lor p_\bot\subseteq p_\theta)$. 
		From this \eqref{eq:boxed} follows.
		\item For each $C_i$, $1\leq i\leq k$, and every $w,w'\in S$ it holds that if $w(\vec{c_i})=w'(\vec{c_i})$ then $w(q_i)=w'(q_i)$. 
		Let $f_1,\dots,f_k$ be some functions that arise from the fact that the constraints $(C_1,\dots,C_k)$ are satisfied in $S$. 
		Since, by assumption, $\varphi$ is non-valid under the constraint $(C_1,\dots,C_k)$, it follows that there exists an assignment $s\colon\{p_1,\dots,p_n\}\to\{0,1\}$ such that \eqref{eq:modified} holds. 
		Now recall that each tuple of values assigned to $\vec p\dfn (p_1,\dots,p_n)$ is realised in $S$. 
		Accordingly, in particular, $s$ and $s(q_1 \mapsto f_1\big(s(\vec{c}_1)), \dots, q_k \mapsto f_k(s(\vec{c}_k)))$ are realised in $S$. 
		For this reason $\mK,S\models \Diamond^\ell(p_\bot\subseteq p_\theta)$, from which \eqref{eq:boxed} follows in a straightforward manner.
	\end{enumerate}
``$\Leftarrow$'': Assume that $\mK,T\models \Box^n\bigl(\Diamond^\ell(\varphi_{\text{cons}}\lor p_\bot\subseteq p_\theta)\bigr)$ 
 holds for every team pointed Kripke model $\mK,T$ that satisfies the structural properties forced by $\varphi_{\text{struc}}$.
 We need to show that $\varphi$ is non-valid under the constraint $(C_1,\dots,C_k)$. 
 In order to show this, let $f_1,\dots,f_k$ be arbitrary functions with arities that correspond to the constraint $(C_1,\dots,C_k)$. 
 Let $\mK,T$ be a team pointed Kripke model and $S$ a team of $\mK$ such that
\begin{enumerate} 
\item[a)] $\mK,T$ satisfies the structural properties forced by $\varphi_{\text{struc}}$,
\item[b)]  $S$ is obtained from $T$ by evaluating the $n$ $\Box$-symbols,
\item[c)]  $f_i\big(w(\vec{c}_i)\big)= w(q_i)$, for each $w\in S$ and $1\leq i\leq k$. 
 \end{enumerate} 
 It is easy to check that such a model and teams always exist. From the assumption we then obtain that
 \begin{equation}\label{eq:last}
 \mK,S\models \Diamond^\ell(\varphi_{\text{cons}}\lor p_\bot\subseteq p_\theta).
 \end{equation}

But since the values of $q_i$s, by construction, do not violate the constraint $(C_1,\dots,C_k)$, we obtain, with the help of the structural properties, that for \eqref{eq:last} to hold is must be the case that $\mK,S_k\models p_\bot\subseteq p_\theta$, where $S_k$ is some team obtained from $S$ by evaluating the $k$ $\Diamond$-symbols in \eqref{eq:last}. 
But this means that there exists an assignment $s\colon\{p_1,\dots,p_n\}\to\{0,1\}$ such that
	\begin{equation}
	s(q_1 \mapsto f_1\big(s(\vec{c}_1)), \dots, q_k \mapsto f_k(s(\vec{c}_k)))\not\models \theta.
	\end{equation}
Consequently, the claim holds.
\end{proof}

In order to compute $\varphiNonVal$ two binary counters bounded above by $n+k+\ell$ need to be maintained. 
Note that $\log(n+k+\ell)$ is logarithmic with respect to the input length. 
That being the case, the reduction is computable in logspace and the lemma applies.
\end{proof}

The construction in the previous proof works also for strict semantics. In the proof of the claim a small adjustment is needed to facilitate the strict semantics of diamond.
As a result we obtain the following.
\fi
\begin{corollary}\label{cor:val-minc-strict}
$\VAL(\Minc)$ under strict semantics is $\coNEXPTIME$-hard w.r.t.\ $\leqlogm$.
\end{corollary}
While the exact complexities of the problems $\VAL(\Minc)$ and $\VAL(\EMinc)$ remain open, it is easy to establish that the complexities coincide.
\begin{theorem}\label{a:val-eminc-minc-equiv}		
Let $\classFont{C}$ be a complexity class that is closed under polynomial time reductions. 
Then $\VAL(\Minc)$ under lax (strict) semantics in complete for $\classFont{C}$ if and only if $\VAL(\EMinc)$ under lax (strict) semantics in complete for $\classFont{C}$.
\end{theorem}\iflong
\begin{proof}
Let  $\varphi$ be a formula of $\EMinc$ and $k$ the modal depth of $\varphi$. 
Let $\varphi_1,\dots,\varphi_n$ be exactly those subformulae of $\varphi$ that occur as a parameter of some inclusion atom in $\varphi$ and let $p_1,\dots,p_n$ be distinct fresh proposition symbols.
Define
\begin{align*}
\varphi_\text{subst} &\dfn \big(\bigwedge_{0\leq i\leq k} \Box^i \bigwedge_{1\leq j\leq n} (p_j\leftrightarrow \varphi_j) \big), \\
\varphi^* &\dfn \varphi_\text{subst}^\bot \lor  (\varphi_\text{subst} \land \varphi^+),
\end{align*}
where $\varphi_\text{subst}^\bot$ denotes the negation normal form of $\neg \varphi_\text{subst}$ and $\varphi^+$ is the formula obtained from $\varphi$ by simultaneously substituting each $\varphi_i$ by $p_i$. 
It is easy to check that $\varphi$ is valid if and only if the $\Minc$ formula  $\varphi^*$ is. 
Clearly  $\varphi^*$ is computable from  $\varphi$ in polynomial time.
\end{proof}\fi

\section{Conclusion}\label{conclusions}
In this paper we investigated the computational complexity of model checking and validity for propositional and modal inclusion logic in order to complete the complexity landscape of these problems in the mentioned logics. 
In particular we emphasise on the subtle influence of which semantics is considered: strict or lax. 
The model checking problem for these logics under strict semantics is $\NP$-complete and under lax semantics $\P$-complete. 
The validity problem is shown to be $\co\NP$-complete for the propositional strict semantics case. 
For the modal case we achieve a $\coNEXPTIME$ lower bound under lax as well as strict semantics. 
The upper bound is left open for further research. 
It is however easy to establish that, if closed under polynomial reductions, the complexities of $\VAL(\Minc)$ and $\VAL(\EMinc)$, and $\VAL(\Minc_s)$ and $\VAL(\EMinc_s)$ coincide, respectively, see Proposition~\ref{a:val-eminc-minc-equiv}.

\iflong
\section{Related work and further research}\label{sec:further}
Tables~\ref{tbl:sat}--\ref{tbl:val} give an overview of the current state of research for satisfiability, model checking and validity in the propositional and modal team semantics setting for both strict and lax variants. In the tables $\AEXP{\poly}$ refers to alternating exponential time with polynomially many alternations. 
We also identify the unclassified cases open for further research. 
As these tables also mention atoms which have not been considered elsewhere in this paper, we will introduce them shortly:

Let $\vec p$, $\vec q$, and $\vec r$ be tuples of proposition symbols and $q$ a proposition symbol. Then $\dep{\vec{p}, r}$ is a \emph{propositional dependence atom} and $\indep{\vec p}{\vec q}{\vec r}$ is a \emph{conditional independence atom} with the following semantics:
\begin{align*}
X\models \dep{\vec p,q} &\;\Leftrightarrow\;  \forall s,t\in X: s(\vec p)=t(\vec p) \text{ implies }s(q)=t(q).\\
X\models \indep{\vec p}{\vec q}{\vec r} &\;\Leftrightarrow\; \forall s,t\in X: \text{ if }s(\vec p)=t(\vec p), \text{then }
\exists u\in X: u(\vec p\vec q)=s(\vec p\vec q) \text{ and }   u(\vec r)=t(\vec r).
\end{align*}
Intuitively, $\indep{\vec p}{\vec q}{\vec r}$ states that for any fixed value for $\vec p$, $\vec q$ and $\vec r$ are informationally independent. We also consider the contradictory negation $\sim$ in our setting:
\[
X\models\,\sim\!\!\varphi \text{ iff } X\not\models\varphi.
\]
Semantics for these atoms in the modal setting is defined analogously. When $\mathcal{C}$ is a set of atoms, we denote by $\PL(\mathcal{C})$ and $\ML(\mathcal{C})$ the extensions of $\PL$ and $\ML$, in the team semantics setting, by the atoms in $\mathcal{C}$, respectively.

A fruitful direction for future research is to study automatic reasoning in the team semantics setting.
\fi

\iflong

\begin{table}\centering
\begin{tabular}{c@{}cc}\toprule
	& \multicolumn{2}{c}{$\PL$ Satisfiability Problem} \\\cmidrule{2-3}
Operator & strict & lax \\\midrule
$\emptyset$ & \multicolumn{2}{c}{\Vhrulefill\; $\NP$ \cite{coo71a,levin73} \Vhrulefill}\\
$\dep{\cdot}$ & \multicolumn{2}{c}{\Vhrulefill\; $\NP$ \cite{lv13} \Vhrulefill}\\
$\subseteq$ & $\EXPTIME$ \cite{hkmv15personal}
		    & $\EXPTIME$ \cite{hkmv15}\\
$\bot$ & $\NP^\star$
       & $\NP$ \cite{hkvv15}\\
$\sim$ & $\AEXP{\poly}^\star$ 
       & $\AEXP{\poly}$ \cite{hkvv15,hkvv15ext}\\
all& $\AEXP{\poly}^\star$ 
       & $\AEXP{\poly}$ \cite{hkvv15,hkvv15ext}\\\bottomrule
\end{tabular}
\begin{tabular}{cc}\toprule
  \multicolumn{2}{c}{$\ML$ Satisfiability Problem} \\\cmidrule{1-2}
  strict & lax \\\midrule
  \multicolumn{2}{c}{\Vhrulefill\; $\PSPACE$ \cite{lad77} \Vhrulefill}\\
  \multicolumn{2}{c}{\Vhrulefill\; $\NEXPTIME$ \cite{sev09} \Vhrulefill}\\
  $\EXPTIME$ \cite{hkmv15personal} & $\EXPTIME$ \cite{hkmv15}\\
  $\NEXPTIME^\star$  & $\NEXPTIME$ \cite{kms16}\\
  ?    & ?\\
 ?     & ?\\\bottomrule
\end{tabular}

\caption{Complexity of Satisfiability, where $\mathrm{all}=\{\dep\cdot, \subseteq, \bot, \sim \}$. \newline $\star$: Proof for lax semantics works also for strict semantics. \newline ?: No nontrivial result is known.}\label{tbl:sat}
\end{table}

\begin{table}\centering
\begin{tabular}{ccc}\toprule
	& \multicolumn{2}{c}{$\PL$ Model Checking}\\\cmidrule{2-3}
Operator & strict & lax\\\midrule
$\emptyset$ & \multicolumn{2}{c}{\Vhrulefill\; $\NC{1}$ \cite{bus87} \Vhrulefill}\\
$\dep{\cdot}$ & \multicolumn{2}{c}{\Vhrulefill\; $\NP$ \cite{el12}\Vhrulefill}\\
$\subseteq$ & $\NP$ [Thm. \ref{mc-plinc-npc}]
		    & $\P$ [Thm. \ref{mc-plinc-pc}]\\
$\bot$ & $\NP^\star$
       & $\NP$ \cite{hkvv15}\\
$\sim$ & $\PSPACE^\star$
       & $\PSPACE$ \cite{hkvv15,muellerDiss}\\
all& $\PSPACE^\star$
       & $\PSPACE$ \cite{hkvv15,muellerDiss}\\\bottomrule
\end{tabular}
\begin{tabular}{cc}\toprule
\multicolumn{2}{c}{$\ML$ Model Checking}\\\cmidrule{1-2}
strict & lax\\\midrule
\multicolumn{2}{c}{\Vhrulefill\; $\P$ \cite{clemsi86,sc02} \Vhrulefill}\\
\multicolumn{2}{c}{\Vhrulefill\; $\NP$ \cite{el12} \Vhrulefill}\\
$\NP$ [Thm. \ref{mc-minc-npc}] & $\P$ [Thm. \ref{mc-minc-pc}]\\
$\NP^\star$ & $\NP$ \cite{kms16}\\
$\PSPACE^\star$ & $\PSPACE$ \cite{muellerDiss}\\
$\PSPACE^\star$ & $\PSPACE$ \cite{hkvv15}\\\bottomrule
\end{tabular}
\caption{Complexity of Model Checking, where $\mathrm{all}=\{\dep\cdot, \subseteq, \bot, \sim \}$. \newline $\star$: Proof for lax semantics works also for strict semantics.}	\label{tbl:mc}
\end{table}

\begin{table}\centering
\begin{tabular}{@{}c@{}cc@{}}\toprule
	& \multicolumn{2}{c}{$\PL$ Validity Problem}\\\cmidrule{2-3}
Operator & strict & lax\\\midrule
$\emptyset$ & \multicolumn{2}{c}{\Vhrulefill\; $\co\NP$ \cite{coo71a,levin73} \Vhrulefill}\\
$\dep{\cdot}$ & \multicolumn{2}{c}{\Vhrulefill\; $\NEXPTIME$ \cite{virtema14} \Vhrulefill}\\
$\subseteq$ & $\co\NP$ [Thm. \ref{val-plinc-conp}]
		    & $\co\NP$ \cite{hkvv15}\\
$\bot$ & $\in\co\NEXPTIME^\NP{}^\star$
       & $\in\co\NEXPTIME^\NP$ \cite{hkvv15}\\
$\sim$ & $\AEXP{\poly}^\star$
       & $\AEXP{\poly}$ \cite{hkvv15,hkvv15ext}\\
all & $\AEXP{\poly}^\star$
       & $\AEXP{\poly}$ \cite{hkvv15,hkvv15ext}\\\bottomrule
\end{tabular}
\begin{tabular}{@{}cc@{}}\toprule
\multicolumn{2}{c}{$\ML$ Validity Problem}\\\cmidrule{1-2}
strict & lax\\\midrule
\multicolumn{2}{c}{\Vhrulefill\; $\PSPACE$ \cite{lad77} \Vhrulefill}\\
\multicolumn{2}{c}{\Vhrulefill\; $\in\NEXPTIME^{\NP}$ \cite{virtema14} \Vhrulefill}\\
$\coNEXPTIME$-h [Cor. \ref{cor:val-minc-strict}]
		    & $\coNEXPTIME$-h [Lem. \ref{lem:val-minc-lax-lower-bound}]\\
?
       & ?\\
?
       & ?\\
?
       & ?\\\bottomrule
\end{tabular}
\caption{Complexity of Validity, where $\mathrm{all}=\{\dep\cdot, \subseteq, \bot, \sim \}$. Complexity classes refer to completeness results, ``-{\normalfont h}.'' denotes hardness and ``$\in$'' denotes containment. \newline $\star$: Proof for lax semantics works also for strict semantics. \newline ?: No nontrivial result is known.}	\label{tbl:val}
\end{table}
\fi

\bibliographystyle{plainnat}
\bibliography{mincmc}

%% file: arxiv-version.bbl
\begin{thebibliography}{10}

\bibitem{blackburn}
P.~Blackburn, M.~de~Rijke, and Y.~Venema.
\newblock {\em Modal Logic}.
\newblock Cambridge Univ. Press, 2001.

\bibitem{bus87}
S.~R. Buss.
\newblock The {B}oolean formula value problem is in {$\mathsf{ALOGTIME}$}.
\newblock In {\em Proc. 19th STOC}, pages 123--131, 1987.

\bibitem{clemsi86}
E.~Clarke, E.~A. Emerson, and A.~Sistla.
\newblock Automatic verification of finite-state concurrent systems using
  temporal logic specifications.
\newblock {\em ACM ToPLS}, 8(2):244--263, 1986.

\bibitem{coo71a}
S.~A. Cook.
\newblock The complexity of theorem proving procedures.
\newblock In {\em Proc. 3rd STOC}, pages 151--158, 1971.

\bibitem{dukovo16}
A.~Durand, J.~Kontinen, and H.~Vollmer.
\newblock Expressivity and complexity of dependence logic.
\newblock In S.~Abramsky, J.~Kontinen, J.~V{\"a}{\"a}n{\"a}nen, and H.~Vollmer,
  editors, {\em Dependence Logic: Theory and Applications}, pages 5--32. 2016.

\bibitem{el12}
J.~Ebbing and P.~Lohmann.
\newblock Complexity of model checking for modal dependence logic.
\newblock In {\em 38th Proc. {SOFSEM}}, pages 226--237, 2012.

\bibitem{Galliani12}
P.~Galliani.
\newblock Inclusion and exclusion dependencies in team semantics - on some
  logics of imperfect information.
\newblock {\em Ann. Pure Appl. Logic}, 163(1):68--84, 2012.

\bibitem{ghk13}
P.~Galliani, M.~Hannula, and J.~Kontinen.
\newblock Hierarchies in independence logic.
\newblock In {\em Proc. 22nd CSL}, volume~23 of {\em LIPIcs}, pages 263--280,
  2013.

\bibitem{GH13}
P.~Galliani and L.~Hella.
\newblock Inclusion logic and fixed point logic.
\newblock In {\em Proc. 22nd CSL}, LIPIcs, pages 281--295, 2013.

\bibitem{gajo79}
M.~R. Garey and D.~S. Johnson.
\newblock {\em Computers and Intractability, A Guide to the Theory of
  NP-Completeness}.
\newblock Freeman, New York, 1979.

\bibitem{mcvp77}
L.~M. Goldschlager.
\newblock The monotone and planar circuit value problems are log-space complete
  for {P}.
\newblock {\em SIGACT News}, 9:25--29, 1977.

\bibitem{gv13}
E.~Gr{\"{a}}del and J.~V{\"{a}}{\"{a}}n{\"{a}}nen.
\newblock Dependence and independence.
\newblock {\em Studia Logica}, 101(2):399--410, 2013.

\bibitem{hkvv15}
M.~Hannula, J.~Kontinen, J.~Virtema, and H.~Vollmer.
\newblock Complexity of propositional independence and inclusion logic.
\newblock In {\em Proc. 40th MFCS}, pages 269--280, 2015.

\bibitem{hkvv15ext}
M.~Hannula, J.~Kontinen, J.~Virtema, and H.~Vollmer.
\newblock Complexity of propositional logics in team semantics.
\newblock {\em CoRR, extended version of \cite{hkvv15}}, abs/1504.06135, 2015.

\bibitem{hk15}
Miika Hannula and Juha Kontinen.
\newblock Hierarchies in independence and inclusion logic with strict
  semantics.
\newblock {\em J. Log. Comput.}, 25(3):879--897, 2015.

\bibitem{hkmv15}
L.~Hella, A.~Kuusisto, A.~Meier, and H.~Vollmer.
\newblock Modal inclusion logic: Being lax is simpler than being strict.
\newblock In {\em Proc. 40th MFCS}, pages 281--292, 2015.

\bibitem{hkmv15personal}
L.~Hella, A.~Kuusisto, A.~Meier, and H.~Vollmer.
\newblock Satisfiability of modal inclusion logic: Lax and strict semantics.
\newblock 2017.
\newblock Corrected version of \cite{hkmv15}, to appear soon on
  arXiv:1504.06409.

\bibitem{hs15}
L.~Hella and J.~Stumpf.
\newblock The expressive power of modal logic with inclusion atoms.
\newblock In {\em Proc. 6th GandALF}, pages 129--143, 2015.

\bibitem{Hodges97c}
W.~Hodges.
\newblock Compositional semantics for a language of imperfect information.
\newblock {\em Logic Journal of the IGPL}, 5(4):539--563, 1997.

\bibitem{kms16}
J.~Kontinen, J.-S. M{\"u}ller, H.~Schnoor, and H.~Vollmer.
\newblock Modal independence logic.
\newblock {\em Journal of Logic and Computation}, 2016.

\bibitem{lad77}
R.~Ladner.
\newblock The computational complexity of provability in systems of modal
  propositional logic.
\newblock {\em SIAM Journal on Computing}, 6(3):467--480, 1977.

\bibitem{levin73}
L.~A. Levin.
\newblock Universal sorting problems.
\newblock {\em Problems of Inform. Transm.}, 9:265--266, 1973.

\bibitem{lv13}
P.~Lohmann and H.~Vollmer.
\newblock Complexity results for modal dependence logic.
\newblock {\em Studia Logica}, 101(2):343--366, 2013.

\bibitem{muellerDiss}
J.-S. M{\"u}ller.
\newblock {\em Satisfiability and Model Checking in Team Based Logics}.
\newblock PhD thesis, Leibniz University of Hannover, 2014.

\bibitem{Peterson2001}
G.~Peterson, J.~Reif, and S.~Azhar.
\newblock Lower bounds for multiplayer noncooperative games of incomplete
  information.
\newblock {\em Computers \& Math. with Applications}, 41(7-8):957 -- 992, 2001.

\bibitem{sv15b}
K.~Sano and J.~Virtema.
\newblock Characterizing frame definability in team semantics via the universal
  modality.
\newblock In {\em Proc. of WoLLIC 2015}, pages 140--155, 2015.

\bibitem{sv16}
K.~Sano and J.~Virtema.
\newblock Characterizing relative frame definability in team semantics via the
  universal modality.
\newblock In {\em Proc. of WoLLIC 2016}, pages 392--409, 2016.

\bibitem{sc02}
P.~Schnoebelen.
\newblock The complexity of temporal logic model checking.
\newblock In {\em Proc. 4th AiML}, pages 393--436, 2002.

\bibitem{sev09}
M.~Sevenster.
\newblock Model-theoretic and computational properties of modal dependence
  logic.
\newblock {\em Journal of Logic and Computation}, 19(6):1157--1173, 2009.

\bibitem{Stockmeyer:1973}
L.~J. Stockmeyer. and A.~R. Meyer.
\newblock Word problems requiring exponential time(preliminary report).
\newblock In {\em Proc. 5th STOC}, pages 1--9, New York, NY, USA, 1973. ACM.

\bibitem{vaananen07}
J.~V\"a\"an\"anen.
\newblock {\em Dependence Logic}.
\newblock Cambridge University Press, 2007.

\bibitem{virtema14}
J.~Virtema.
\newblock Complexity of validity for propositional dependence logics.
\newblock {\em Information and Computation}, 2016.
\newblock Online first.

\bibitem{vol99}
H.~Vollmer.
\newblock {\em Introduction to Circuit Complexity -- A Uniform Approach}.
\newblock Texts in Theoretical Computer Science. Springer Verlag, Berlin
  Heidelberg, 1999.

\end{thebibliography}
